\documentclass[11pt,letterpaper]{amsart}
\usepackage{graphicx}    % needed for including graphics e.g. EPS, PS
\usepackage{epsfig,amsmath,euscript,amssymb}
\usepackage{amsthm}
\newtheorem{lemma}{Lemma}

\newtheorem{proposition}{Proposition}
\newtheorem{remark}{Remark}
\newtheorem{assumption}{Assumption}
\newtheorem{corollary}{Corollary}

\begin{document}

\title[The quasi-Gaussian log-normal HJM model]
{Small-noise limit of the quasi-Gaussian log-normal HJM model}

\author{Dan Pirjol}
\email{
dpirjol@gmail.com}

\author{Lingjiong Zhu}
\email{
ling@cims.nyu.edu}

\date{}

%\subjclass[2010]{91B70,60J60,34D20} 

\begin{abstract}
Quasi-Gaussian HJM models are a popular approach for modeling 
the dynamics of the yield curve. This is due to their low dimensional 
Markovian representation, which greatly simplifies their numerical 
implementation. We present a qualitative study of the solutions of the
quasi-Gaussian log-normal HJM model. Using a small-noise deterministic
limit we show that the short rate may explode to infinity in finite 
time. This implies the explosion of the Eurodollar futures prices in this model.
We derive explicit explosion criteria under mild assumptions on the
shape of the yield curve. 
\end{abstract}

\maketitle

\section{Introduction}

HJM models \cite{HJM} are widely used in financial practice for modeling 
fixed income, credit and commodity markets \cite{APbook}. These models specify 
the dynamics of the yield curve $f(t,T)$ as 
\begin{equation}
df(t,T) = \sigma_f(t,T)^T dW(t) + 
\sigma_f^T(t,T) \left( \int_t^T \sigma_f(t,s) ds \right) dt \,,
\end{equation}
where $W(t)$ is a vector Brownian motion under the risk-neutral measure 
$\mathbb{Q}$ and $\{\sigma_f(t,T)\}_{t\leq T}$ is a family of vector processes.
The numerical simulation of these models is complicated by the fact that the 
entire yield curve $f(t,T)$ has to be simulated. Lattice and tree simulation 
methods require an exponentially large number of nodes. 
For this reason the simulation of these models is restricted in practice to 
Monte Carlo methods. 

The quasi-Gaussian HJM models \cite{APbook,review,B,C,RS} were introduced to 
simplify the simulation of the HJM models.
They are obtained by assuming a separable form for the volatility 
$\sigma_f(t,T)^T = g(T)^T h(t)$ where $g$ is a deterministic vector function
and $h$ is a $k\times k$ matrix process. Such models admit a Markov 
representation of the dynamics of the yield curve involving $k+\frac12 k(k+1)$ 
state variables. This simplifies very much their simulation, which can be 
done either using Monte Carlo or finite difference methods \cite{CaZhu,PhD}. 

We consider in this note the one-factor quasi-Gaussian HJM model with
volatility specification $\sigma_f(t,T) = k(t,T) \sigma(r_t)$ where 
$k(t,T)=e^{-\beta(T-t)}$, and $\sigma(r_t)$ is the volatility of the short rate
$r_t = f(t,t)$. This model admits a two state Markov representation.

It has been noted in \cite{Morton,HJM} that in HJM models with log-normal
volatility specification, that is for which $\sigma_f(t,T) = \sigma(t,T)
f(t,T)$, the rates explode to infinity with probability one, and 
zero coupon bond prices are zero. 
It is natural to ask if a similar explosive phenomenon is present also in the
quasi-Gaussian HJM model with log-normal volatility $\sigma(r_t) = \sigma r_t$. 
This model is used in financial practice for modeling swaption volatility smiles
\cite{Chibane} and is a particular case of a more general parametric representation
\cite{quadratic}.

We study in this note the qualitative behavior of the solutions of this model.
In the small-noise deterministic limit, we show rigorously that the short rate may explode to infinity in a finite time. 
More precisely, for sufficiently small mean-reversion $\beta$, 
the deterministic approximation for the short rate has an explosion in 
finite time, and an upper bound is given on the explosion time, which
is saturated in the flat forward rate limit.
When Brownian noise is taken into account, the explosion time has a
distribution around the deterministic limit. 

This phenomenon has implications for the practical use of the model
for pricing and simulation. It implies an explosion of the Eurodollar futures 
prices in this model, and introduces a limitation in the applicability of the 
model for pricing these products to maturities smaller than the explosion time. 

\section{Log-normal quasi-Gaussian HJM model}

The one-factor log-normal quasi-Gaussian HJM model is defined by the 
volatility specification
\begin{equation}
\sigma_f(t,T) = \sigma r_t e^{-\beta (T-t)} \,.
\end{equation}
The simulation of the model requires the solution of the stochastic
differential equation for the two variables $\{ x_t,y_t\}_{t\geq 0}$ 
\cite{RS,APbook}
\begin{eqnarray}\label{xySDE}
&& dx_t = (y_t - \beta x_t) dt + \sigma (\lambda(t) + x_t) dW_t  ,\\
&& dy_t = (\sigma^2 (\lambda(t) + x_t)^2 - 2\beta y_t) dt ,\nonumber
\end{eqnarray}
with initial condition $x_0 = y_0 = 0$. Here $\lambda(t) = f(0,t)$ is
the forward short rate, giving the initial yield curve. 
The zero coupon bonds are
\begin{equation}\label{ZCB}
P(t,T) = \frac{P(0,T)}{P(0,t)}
\exp\left(- G(t,T) x_t - \frac12 G^2(t,T) y_t \right)\,,
\end{equation}
with $G(t,T) \geq 0$ a non-negative deterministic function \cite{APbook}. 
The short rate is $r_t := f(t,t) = \lambda(t) + x_t$.
The equations (\ref{xySDE}) can be expressed in terms of the short rate as
\begin{eqnarray}\label{rySDE}
&& dr_t = (y_t - \beta r_t + \beta \lambda(t) + \lambda'(t)) dt 
        + \sigma r_t dW_t ,\\
&& dy_t = (\sigma^2 r_t^2 - 2\beta y_t)  dt \,,\nonumber
\end{eqnarray}
with the initial condition $r_0 = \lambda_0 := \lambda(0) > 0$ and $y_0 = 0$.

The solutions of the process \eqref{rySDE} may explode with 
non-zero probability. This will be discussed in \cite{PZ}. 
When the volatility $\sigma=0$, there is no explosion. 
Indeed, when $\sigma=0$, we have
\begin{align*}
&dr_{t}=(y_{t}-\beta r_{t}+\beta\lambda(t)+\lambda'(t))dt,
\\
&dy_{t}=-2\beta y_{t}dt,
\end{align*}
with the initial condition $r_{0}=\lambda_{0}$ and $y_{0}=0$.
Thus $y_{t}\equiv 0$, which gives 
$r'_{t}=-\beta r_{t}+\beta\lambda(t)+\lambda'(t)$.
This ODE can be easily solved with the result $r_t = \lambda(t)$.

%%%%%%%%%%%%%%%%%%%%%%%%%%%%%%%%%%%%%%%%%%%%%%%%
\section{Deterministic approximation}
\label{sec:3}

Instead of studying directly the distribution of the explosion time of the
process $(r_t,y_t)$, we study a deterministic proxy of the equations 
(\ref{rySDE}). 
In the limit when the Brownian noise in these equations goes to zero, then
$(r_{t},y_{t})\rightarrow(r(t),y(t))$, where $(r(t),y(t))$ satisfy the 
two-dimensional ODE:
\begin{eqnarray}\label{ODEr}
&& r'(t)=y(t)-\beta r(t)+\beta\lambda(t)+\lambda'(t), \\
&& y'(t)=\sigma^{2} r^2(t)-2\beta y(t), \nonumber
\end{eqnarray}
with $r(0)=\lambda_{0}$ and $y(0)=0$.
The variable $r(t)$ can be interpreted as the deterministic approximation of the short rate $r_{t}$
and its expected value $\mathbb{E}^{\mathbb{Q}}[r_{t}]$ for the small-noise limit.
The pair $(r(t),y(t))$ is a deterministic approximation of the two-dimensional SDE \eqref{rySDE}.

We study here the qualitative properties of the solution for $r(t)$.
Even though (\ref{ODEr}) is a system of 2-dim ODEs, 
we will show that $r(t)$ can be expressed as a solution to a 1-dim
integral equation.

\begin{proposition}\label{Prop:1}
$r(t)$ satisfies 
the integral equation
%a one-dimensional time-inhomogeneous integro-differential equation:
%\begin{equation}\label{drdt}
%dr(t)=\left(\sigma^{2}\int_{0}^{t}r^2(s)
%e^{2\beta(s-t)}ds-\beta r(t)+\beta\lambda(t)+\lambda'(t)\right)dt.
%\end{equation}
%Moreover, 
\begin{equation}\label{integrodiffeq}
r(t)=\lambda(t)+\frac{\sigma^{2}}{\beta}
\int_{0}^{t}r^2(s) [e^{\beta(s-t)}-e^{2\beta(s-t)}]ds.
\end{equation}
\end{proposition} 

\begin{proof}
We can solve for $y(t)$ as
\begin{equation}\label{ysol}
y(t)=\sigma^{2}\int_{0}^{t}r^2(s) e^{2\beta(s-t)}ds.
\end{equation}
%Therefore, we get \eqref{drdt}.
%Rearranging \eqref{drdt}, we get
Substituting into (\ref{ODEr}) we get
\begin{equation}
r'(t)+\beta r(t)=\sigma^{2}\int_{0}^{t}r^2(s) e^{2\beta(s-t)}ds
+\beta\lambda(t)+\lambda'(t).
\end{equation}
Multiplying by the integrating factor $e^{\beta t}$ and integrating from 
$0$ to $t$, we obtain:
\begin{align}
r(t)e^{\beta t}-\lambda_{0}
&=\sigma^{2}\int_{0}^{t}\int_{0}^{u}r^2(s) e^{2\beta s}e^{-\beta u}du ds
+\lambda(t)e^{\beta t}-\lambda_{0}
\\
&=\lambda(t)+\sigma^{2}\int_{0}^{t}\int_{s}^{t}
  r^2(s) e^{2\beta s}e^{-\beta u}du ds
-\lambda_{0} \,, \nonumber
\end{align}
which yields Eq.~\eqref{integrodiffeq}.
\end{proof}

We show next that if $\lambda(t)$ is uniformly bounded, 
for sufficiently large $\beta$ or sufficiently small $\sigma$, $r(t)$ 
is also uniformly bounded, and hence there will be no explosion.

\begin{proposition}\label{Prop:2}
Assume that $\lambda(t)$ is uniformly bounded. 
Then, for sufficiently large $\beta$
or sufficiently small $\sigma$, we have
\begin{equation}
\max_{t\geq 0}r(t)
\leq\frac{\beta^{2}}{\sigma^{2}}\left(1-\sqrt{1-\max_{t\geq 0}\lambda(t)\frac{2\sigma^{2}}{\beta^{2}}}\right).
\end{equation}
It follows that there will be no explosion.
\end{proposition}

\begin{proof}
We only give a proof for the large $\beta$ result. 
The same result holds for sufficiently small $\sigma$ with a very similar proof.

For any $t\in[0,T]$, we have from Eq.~(\ref{integrodiffeq})
\begin{equation}
r(t)\leq\max_{0\leq t\leq T}\lambda(t)
+\left[\max_{0\leq t\leq T}r(t)\right]^{2}\frac{\sigma^{2}}{\beta}\int_{0}^{\infty}[e^{-\beta s}-e^{-2\beta s}]ds,
\end{equation}
which implies that $R(T):=\max_{0\leq t\leq T}r(t)$ satisfies
\begin{equation}
R(T)-R^2(T)\frac{\sigma^{2}}{2\beta^{2}}\leq\max_{0\leq t\leq T}\lambda(t).
\end{equation}
This implies that we have either 
(i) $R(T) \leq R_{1}(T)$, or (ii) $R(T) \geq R_{2}(T)$, with
\begin{equation}
R_{1,2}(T) := \frac{\beta^{2}}{\sigma^{2}}
\left(1\mp \sqrt{1-\max_{0\leq t\leq T}\lambda(t)\frac{2\sigma^{2}}{\beta^{2}}}
\right)\,.
%\qquad R_1(T) \leq R_2(T)\,.
\end{equation}
For large $\beta$, $R(T)$ is bounded by Proposition~\ref{Prop:2},
while $R_2(T)\to \infty$ as $\beta\rightarrow\infty$. 
Therefore, for sufficiently large $\beta>0$, we have $R(T) \leq R_1(T)$.
Taking now $T\to \infty$, we have by the uniformly bounded assumption 
$\max_{t\geq 0}\lambda(t)<\infty$. It follows that for sufficiently large 
$\beta$,
\begin{equation}
\max_{t\geq 0}r(t)
\leq\frac{\beta^{2}}{\sigma^{2}}
\left(1-\sqrt{1-\max_{t\geq 0}\lambda(t)\frac{2\sigma^{2}}{\beta^{2}}}\right).
\end{equation}
We conclude that for sufficiently large $\beta$, $r(t)$ is not explosive and 
is indeed uniformly bounded as long as $\lambda(t)$ is uniformly bounded.
\end{proof}

\begin{remark}\label{uniformRemark}
From Proposition \ref{Prop:2}, it follows
that $\max_{t\geq 0}r(t)$ is uniformly bounded as either
$\beta\rightarrow\infty$ or $\sigma\rightarrow 0$, since
\begin{equation*}
\limsup_{\beta\rightarrow\infty}\max_{t\geq 0}r(t)
\leq\limsup_{\beta\rightarrow\infty}\frac{\beta^{2}}{\sigma^{2}}\left(1-\sqrt{1-\max_{t\geq 0}\lambda(t)\frac{2\sigma^{2}}{\beta^{2}}}\right)
=\max_{t\geq 0}\lambda(t),
\end{equation*}
and the same result holds for $\sigma\rightarrow 0$.
\end{remark}

From Proposition \ref{Prop:1}, Proposition \ref{Prop:2} and Remark \ref{uniformRemark}, we immediately 
get the following corollary.

\begin{corollary}\label{corollary1}
As $\beta\rightarrow\infty$ (resp. $\sigma\rightarrow 0$), we have $r(t)\rightarrow\lambda(t)$ uniformly for $t\geq 0$.
More precisely, 
\begin{equation*}
\max_{t\geq 0}|r(t)-\lambda(t)|
\leq\left[\max_{t\geq 0}\lambda(t)\right]^{2}\frac{\sigma^{2}}{2\beta^{2}}+o(\beta^{-2})
\quad(\text{resp. } o(\sigma^{2})),
\qquad\text{as $\beta\rightarrow\infty$ (resp. $\sigma\rightarrow 0$)}.
\end{equation*}
\end{corollary}

\begin{proof}
From Proposition \ref{Prop:1}, Proposition \ref{Prop:2} and Remark \ref{uniformRemark}, we get
\begin{align*}
\max_{t\geq 0}|r(t)-\lambda(t)|
&\leq\max_{t\geq 0}\frac{\sigma^{2}}{\beta}
\int_{0}^{t}r^2(t-s)[e^{-\beta s}-e^{-2\beta s}]ds
\\
&\leq\left[\max_{t\geq 0}r(t)\right]^{2}\frac{\sigma^{2}}{\beta}\int_{0}^{\infty}[e^{-\beta s}-e^{-2\beta s}]ds
\\
&\leq\left[\max_{t\geq 0}\lambda(t)\right]^{2}\frac{\sigma^{2}}{2\beta^{2}}+o(\beta^{-2})
\quad(\text{resp. } o(\sigma^{2})),
\quad\text{as $\beta\rightarrow\infty$ (resp. $\sigma\rightarrow 0$).}
\nonumber
\end{align*}
\end{proof}

We can also study the stationary limit for the deterministic system, that is, 
the large time limit for the deterministic system. 

\begin{proposition}
Let us assume that $\lim_{t\rightarrow\infty}\lambda(t)=\lambda(\infty)$ and $\lim_{t\rightarrow\infty}r(t)=r(\infty)$ exist, then
\begin{equation}\label{rinfty}
r(\infty)=\frac{\beta^{2}}{\sigma^{2}}\left(1-\sqrt{1-\lambda(\infty)\frac{2\sigma^{2}}{\beta^{2}}}\right).
\end{equation}
\end{proposition}

\begin{proof}
Let us recall that
\begin{equation}
r(t)=\lambda(t)+\frac{\sigma^{2}}{\beta}\int_{0}^{t}(r(t-s))^{2}
[e^{-\beta s}-e^{-2\beta s}]ds.
\end{equation}
Thus the large time limit $r(\infty)$ is the smaller root of the equation
\begin{equation}
r(\infty)=\lambda(\infty)+r(\infty)^{2}\frac{\sigma^{2}}{\beta}\int_{0}^{\infty}[e^{-\beta s}-e^{-2\beta s}]ds,
\end{equation}
which yields \eqref{rinfty}.
\end{proof}

For constant $\lambda(t)\equiv\lambda_{0}$, one can study the 
stability of the two-dimensional first-order ODE
\begin{eqnarray}
r'(t) &=& y(t) - \beta r(t) + \beta \lambda_0 ,\\
y'(t) &=& \sigma^2 r^2(t) - 2\beta y(t) .
\end{eqnarray}

We would like to determine the fixed points of this equation, and determine 
their type.
The fixed points are given by the zeros of the functions on the RHS of the ODE. 
For $\beta^2 < \beta_C^2 := 2\lambda_0 \sigma^2$ there are no fixed points. 
For $\beta^2 = \beta_C^2$ there is one at $r_1=2\lambda_0$, 
$y_1 = \lambda_0 \beta_C$.
For $\beta^2 > \beta_C^2$ there are two fixed points:
\begin{eqnarray}
\Pi_{1,2} &:& r_{1,2} = \frac{\beta^2}{\sigma^2}
 \left(1 \mp \sqrt{\Delta}\right)\,,\quad
y_{1,2} =  \frac{\beta^3}{2\sigma^2}
 \left(1 \mp \sqrt{\Delta}\right)^2 \,.
%\Pi_2 &:& r_2 = \frac{\beta^2}{\sigma^2}
% \left(1 + \sqrt{\Delta}\right)\,,\quad
%y_2 =  \frac{\beta}{2}
% \left(1 + \sqrt{\Delta}\right) \,.
\end{eqnarray}
with $\Delta := 1 - \frac{2\lambda_0\sigma^2}{\beta^2}$.
Linearization of the ODE around each fixed point gives the linear ODE
\begin{equation}
\left( 
\begin{array}{c}
r'(t) \\
y'(t) \\
\end{array}
\right) = 
\left( 
\begin{array}{cc}
- \beta & 1 \\
2\sigma^2 r_i & -2\beta \\
\end{array}
\right).
\left( 
\begin{array}{c}
r(t) - r_i \\
y(t) - y_i \\
\end{array}
\right) + \mbox{ higher order terms }\,.
%A_i
%\left( 
%\begin{array}{c}
%r(t) - r_i \\
%y(t) - y_i \\
%\end{array}
%\right) + \mbox{ higher order terms }.
\end{equation}

%The matrix $A_i$ is
%\begin{equation}
%A_i = \left( 
%\begin{array}{cc}
%- \beta & 1 \\
%2\sigma^2 r_i & -2\beta \\
%\end{array}
%\right).
%\end{equation}

%The relative signs of the eigenvalues determine the type of fixed point:

%1. $\lambda_1, \lambda_2 >0$. This is a repelling fixed point.

%2. $\lambda_1, \lambda_2 < 0$. This is an attracting fixed point.

%3. $\lambda_1>0, \lambda_2 < 0$ or vice versa. This is a saddle point.

The eigenvalues of the matrix of coefficients for each fixed point are:
\begin{eqnarray}
\Pi_1 &:& \lambda_{1,2} = \frac12 \left(-3\beta \pm 3\beta 
 \sqrt{1 - \frac89 \sqrt{\Delta}} \right),\\
\Pi_2 &:& \lambda_{1,2} = \frac12 \left(-3\beta \pm 3\beta 
 \sqrt{1 + \frac89 \sqrt{\Delta}} \right) \,.
\end{eqnarray}

The relative signs of the two eigenvalues 
determine the type of the fixed points: $\Pi_1$ is an attractive fixed 
point ($\lambda_1, \lambda_2 < 0$), and 
$\Pi_2$ is a saddle point ($\lambda_1>0, \lambda_2 < 0$). 
This means that the two-dimensional ODE describes a flow for $(r(t), y(t))$,
and the lines of flow can either end at $\Pi_1$, or at infinity, avoiding $\Pi_2$.

\section{Explosion criteria}

We study in this Section in more detail the explosion time of $r(t)$, 
defined as $\tau_\infty := \sup\{\tau : r(\tau) < \infty \} $.
The starting point of the analysis is the 2nd order ODE for $r(t)$
\begin{equation}\label{ODEfull}
r''(t) + 3\beta r'(t) + 2\beta^2 r(t) =
\sigma^2 r^2(t) + 2\beta^2 \lambda_0 + \Lambda(t), 
\end{equation}
where we defined $\Lambda(t):= 2\beta^2 (\lambda(t) - \lambda(0))
+ 3\beta \lambda'(t) + \lambda''(t)$.

We will make the following assumption about the initial forward rate $\lambda(t)$.
\begin{assumption}\label{assump1}
Assume that $\lambda(t)$ satisfies the condition
\begin{equation}
\Lambda(t) := 2\beta^2 (\lambda(t) - \lambda(0))
+ 3\beta \lambda'(t) + \lambda''(t) \geq 0\,,\qquad  t \geq 0\,.
\end{equation}
\end{assumption} 
This is satisfied by forward rate curves $\lambda(t)$ which are flat or
up-sloping and not too concave, which are usual in normal market conditions. 

Under this assumption, $r(t)$ satisfies the differential inequality,
which reduces to equality in the limit of a constant $\lambda(t)=\lambda_0$
\begin{equation}\label{ODEfull3}
r''(t) + 3\beta r'(t) + 2\beta^2 r(t) \geq
\sigma^2 r^2(t) + 2\beta^2 \lambda_0 ,
\end{equation}
with initial conditions  
\begin{equation}\label{ic}
r(0) = \lambda_0,\qquad r'(0) = 0\,.
\end{equation}

\subsection{Zero mean-reversion case $\beta=0$}

\begin{proposition}\label{prop:beta0}
The explosion time of $r(t)$ in the one-factor log-normal quasi-Gaussian 
HJM model under the Assumption~\ref{assump1} with $\beta=0$ is bounded from 
above as
\begin{equation}\label{ineqeq}
\tau_\infty \leq \frac{\sqrt{6p_0} \omega_2}{\sigma \sqrt{\lambda_0}} = 
\frac{2.97448}{\sigma \sqrt{\lambda_0}},
\end{equation}
where $p_0=0.62996$, $\omega_2=1.52995$ are parameters of the 
Weierstrass elliptic function \cite{WW}. The inequality in \eqref{ineqeq} 
becomes sharp when $\lambda(t)\equiv\lambda_0$.
\end{proposition}

\begin{proof}
Taking $\beta=0$ in (\ref{ODEfull}) we get the differential
inequality
\begin{equation}\label{ODE}
r''(t) \geq \sigma^2 r^2(t)  \,.
\end{equation}
with initial conditions (\ref{ic}).

By Lemma~\ref{lemma:1} (see Appendix) it is sufficient to study the solution 
of the ODE obtained by replacing the inequality sign in (\ref{ODE}) with
equality, and satisfying the same initial condition at $t=0$. 
The solution of this equation $Z''(t)=\sigma^{2}Z^{2}(t)$
with $Z(0)=\lambda_{0}$, $Z'(0)=0$ can be found exactly in 
terms of the Weierstrass elliptic function $\wp(z;c_1,c_2)$, and is given by 
\begin{equation}\label{ZODE}
Z(t) = \frac{6^{1/3}}{\sigma^{2/3}}
\wp\left(\left(\frac{\sigma^2}{6}\right)^{1/3}(t + c_1); 0,c_2\right),
\end{equation}
where the constants $c_1,c_2$ are given by
\begin{equation}\label{c12sol}
c_1 = \sqrt{6p_0} \omega_2 \frac{1}{\sigma \sqrt{\lambda_0}}\,,\qquad
c_2 = \frac{1}{6p_0^3} \lambda_0^3 \sigma^2 ,
\end{equation}
with $p_0=0.62996$, $\omega_2=1.52995$.
This can be simplified further by using the relation \cite{WW}
$\wp(z;0,c_2) = c_2^{1/3} \wp(z c_2^{1/6}; 0,1)$
to rescale the second argument $c_2$ to 1.

\begin{figure}
    \centering
   \includegraphics[width=2.8in]{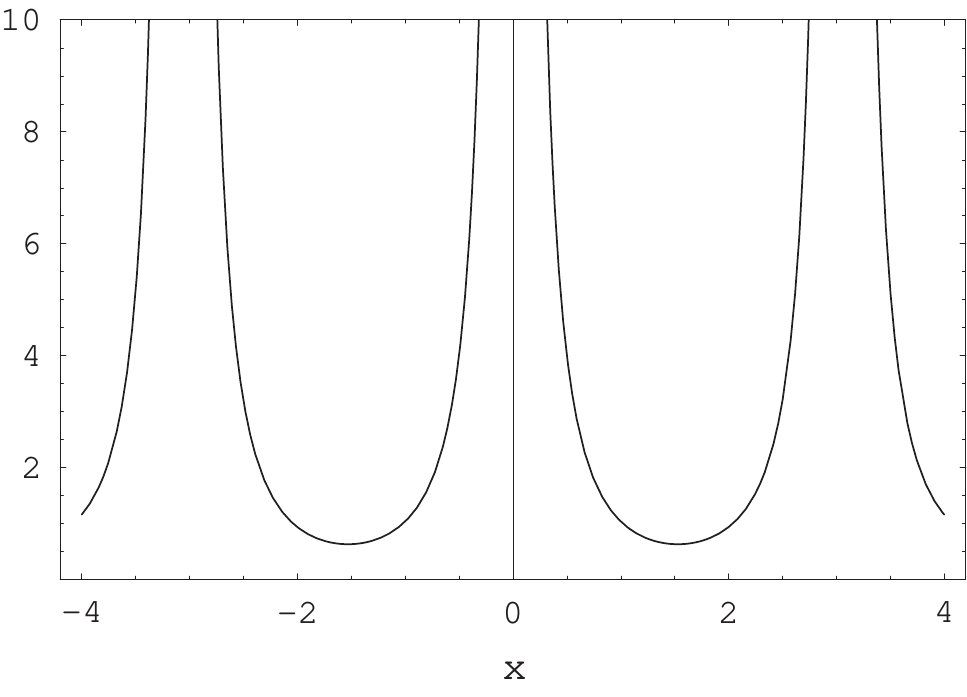}
    \caption{ Plot of the Weierstrass elliptic function $\wp(x;0,1)$. }
\label{Fig:wp}
 \end{figure}

The main features of the solution can be understood by recalling the main 
properties of $\wp(z;0,1)$ \cite{WW}. Along the real axis this function is 
periodic with
half-period $\omega_2 \simeq 1.52995$, and has double poles at origin $z=0$
and at all points uniformly spaced by $2\omega_2$. Half-way between poles at
$z=\omega_2$ it reaches a minimum value $p_0=0.62996$. The plot of 
$\wp(z;0,1)$ is shown in Fig.~\ref{Fig:wp}.

Time zero corresponds to $z=\omega_2$ and the explosion of $r(t)$ takes place at the 
nearest pole, at $z = 2\omega_2$. This gives for the explosion time of the
solution $Z(t)$, defined as $\bar \tau_\infty := \sup\{ t : Z(t) < \infty\}$
\begin{equation}\label{texp}
\bar \tau_\infty = c_1 = \frac{1}{\sigma \sqrt{\lambda_0}} \sqrt{6p_0} \omega_2 =
2.97448 \frac{1}{\sigma \sqrt{\lambda_0}}\,.
\end{equation}
For constant forward rate $\lambda(t)=\lambda_0$ we have 
$\tau_\infty = \bar \tau_\infty$. Allowing for non-constant $\lambda(t)$ 
satisfying Assumption~\ref{assump1} we have the upper bound on the
explosion time $\tau_\infty \leq \bar \tau_\infty$.
%This concludes the proof of the Proposition~\ref{prop:beta0}.
\end{proof}

As a numerical illustration 
%for the result of Proposition~\ref{prop:beta0}
we consider the following typical model parameters: constant forward rate 
$\lambda_0 = 5\%$ and volatility $\sigma=20\%$.
The explosion time of the short rate $r(t)$ is 
$\tau_\infty=66.5Y$. This decreases as the volatility $\sigma$
increases or the forward rate $\lambda_0$ increases.

%%%%%%%%%%%%%%%%%%%%%%%%%%%%%%%%%%
\subsection{Positive mean-reversion case $\beta>0$}

It was shown in Proposition~\ref{Prop:2} that in the
limit $\beta \to \infty$ there is no explosion in $r(t)$. On the other hand,
for $\beta=0$ there is explosion (Proposition~\ref{prop:beta0}). This implies that as $\beta$ increases
from zero, there is a maximum value of $\beta$ for which there is explosion.
For constant $\lambda(t) = \lambda_0$ this maximum value can be found explicitly,
and is given by the following result. For $\beta$ less than this
maximum value, the result can be extended to a wider class of 
functions $\lambda(t)$, satisfying Assumption~\ref{assump1}.

\begin{proposition}\label{Prop:5}
The critical value $\beta_C:= \sigma\sqrt{2\lambda_0}$ separates the 
solutions of (\ref{ODEfull}) 
into two families, with different qualitative behavior:

(i) Small mean-reversion $\beta < \beta_{\rm C}$. Under the Assumption~\ref{assump1}
the explosion time of the solution for $r(t)$ is bounded from above as
$\tau_\infty \leq \bar \tau_\infty$ with
\begin{equation}\label{texp2}
\bar\tau_\infty = \int_{\lambda_0}^\infty \frac{dx}{\sqrt{y(x)}},
\end{equation}
where $y(x)$ is the solution of the first order nonlinear ODE
\begin{equation}\label{yeq}
\frac12 y'(x) + 3\beta \sqrt{y(x)} = 
 \sigma^2 x^2 - 2\beta^2 x + 2\beta^2 \lambda_0 \,,
\end{equation}
with initial condition $y(\lambda_0)=0$. The inequality 
$\tau_\infty \leq \bar \tau_\infty$ becomes sharp for
$\lambda(t) = \lambda_0$.

(ii) Large mean-reversion $\beta \geq \beta_{\rm C}$. Assuming constant
$\lambda(t) = \lambda_0$, the solution for $r(t)$ does not reach infinity.
It satisfies $\lim_{t\to \infty} r(t) = x_1$ with $x_1$ given by (\ref{x12sol}).
For general $\lambda(t)$ a numerical study is required in order to
decide the absence or presence of an explosion for $r(t)$.  
\end{proposition}

\begin{proof}
We start by assuming $\lambda(t) = \lambda_0$, and prove the equation
(\ref{yeq}).
For sufficiently small $\beta$,  $r(t)$ is a strictly increasing function 
of $t$ and can be inverted.
Denote $v(x) := r'(r^{-1}(x))$. It is easy to see that this satisfies 
\begin{equation}
v'(x) = \frac{d}{dx} v(x) = r''(t) \frac{1}{r'(t)}.
\end{equation}
Multiplying with $v(x)$ and using (\ref{ODEfull}) we get that this
function satisfies the differential equation
\begin{equation}\label{veq}
v(x) v'(x) + 3\beta v(x) = \sigma^2 x^2 - 2\beta^2 x + 2\beta^2 \lambda_0 \,.
\end{equation}
(With the substitution $z=-3\beta x$, 
the equation (\ref{veq}) can be brought into the canonical form of 
the Abel equation of the second kind $v(z) v'(z) - v(z) = f(z)$.)
This equation is difficult to simulate numerically as its behavior 
near $x=\lambda_0$ is singular, 
i.e. $\lim_{x\to \lambda_0} v'(x) = \infty$. 
It is more convenient to define $y(x) = v^2(x)$, which
satisfies (\ref{yeq}) and has well-behaved behavior $\lim_{x\to \lambda_0}
y'(x) = 2\sigma^2\lambda_0^2 $. This reproduces Eq.~(\ref{yeq}).

We would like to study the finiteness of the integral in (\ref{texp2}).
For sufficiently small positive $\beta$, the function $v(x)=\sqrt{y(x)}$ is everywhere
positive. Its large $x$ asymptotics is $v(x) = c_1 x^{3/2} + O(x^{1/2})$, 
so the integral  (\ref{texp2}) converges at the upper limit of integration.
However, we will show that as $\beta$ increases, the solution $y(x)$ develops
a zero at a point $x_0$ as $\beta$ approaches a certain value 
$\beta_{\rm C}$. At this point the integral (\ref{texp2}) diverges.

Let us consider the behavior of the solution $v(x)$ as $\beta$ 
is increased from zero. 
Denote $f(x) := \sigma^2 x^2 - 2\beta^2 x  + 2\beta^2\lambda_0$
the function on the right hand side of (\ref{yeq}) (with $\Lambda(t)=0$). 
$f(x)$ is always positive for $\beta < \beta_C = \sigma\sqrt{2\lambda_0}$ 
and becomes negative in a region $(x_1,x_2)$ for $\beta > \beta_C$, with
\begin{equation}\label{x12sol}
x_{1,2} = \frac{\beta^2}{\sigma^2}
\left(1 \pm \sqrt{1 - \frac{2\sigma^2\lambda_0}{\beta^2}}\right)\,.
\end{equation}
For $\beta=\beta_C$ the function $f(x)$ vanishes at $x_{1}=x_{2}=2\lambda_0$.
These solutions satisfy the inequalities $x_{1,2} > \lambda_0$. 
As $\beta \to \infty$, the smallest solution approaches $\lambda_0$ from above: 
$\lim_{\beta \to \infty} x_1 = \lambda_0$. 
The points $x_{1,2}$ are fixed points for the equation (\ref{ODEfull}),
and correspond to the fixed points $\Pi_{1,2}$ for the system $(r(t),y(t))$
studied in Section~\ref{sec:3}.
Recall that $x_1$ is a stable fixed point, and $x_2$ is a saddle point. 

We note the following properties of the solution $y(x)$ of the equation (\ref{yeq}):

(1) $y(x)$ is an increasing function in a neighborhood of $x=\lambda_0$.
This follows from the fact that for any $\beta \geq 0$ we have 
$f(\lambda_0) > 0$. 

(2) For sufficiently large $\beta > \beta_C$, the solution must 
be a decreasing function in the region $(x_1,x_2)$. This follows from
\begin{equation}
\frac12 y'(x)  = f(x) - 3\beta \sqrt{y(x)} \leq f(x) < 0 \,,\quad
x_1 \leq x \leq x_2\,.
\end{equation}

From this analysis we conclude that the solution $y(x)$ has a minimum 
for sufficiently large $\beta$. Denote $x_0$ the position of this minimum. 
Let us study the behavior of $y(x)$ in the neighborhood of $x_0$. 
From (\ref{yeq}) we have
\begin{equation}
3\beta \sqrt{y(x_0)} = f(x_0) \,.
\end{equation}
From the properties of $f(x)$ discussed above, we see that as $\beta$ increases, 
$y(x_0)$ approaches zero as $\beta \to \beta_{\rm C}$. 
Thus $\lim_{\beta \to \beta_{\rm C}}
x_0 = 2\lambda_0$ must be the double root of $f(x)=0$ at $\beta=\beta_C$. 

Taking one derivative of (\ref{yeq}) with respect to $x$ we have
\begin{equation}
\frac12 y''(x) + \frac{3\beta}{2\sqrt{y(x)}} y'(x) = 2\sigma^2 x - 2\beta^2 .
\end{equation}
At $x=x_0$, we have $y'(x_0)=0$ which gives
\begin{equation}
\lim_{\beta \to \beta_0} \frac12 y''(x_0) = 2\sigma^2 x_0 - 2\beta^2 \,.
\end{equation}

Thus we get the expansion of the solution $y(x)$ around the minimum at $x_0$
\begin{equation}\label{yasympt}
y(x) =y(x_0) + 2(\sigma^2 x_0 - \beta^2) (x - x_0)^2 + O((x-x_0)^3),
\end{equation}
with $y(x_0) = f^2(x_0)/(3\beta)^2$. 
This shows that as $\beta \to \beta_{\rm C} $, the 
integral (\ref{texp2}) diverges due to the singularity at $x=x_0$.
\end{proof}

The qualitative behavior described above is seen in Figure~\ref{Fig:vcr} 
which shows numerical solutions for $v(x)$ for three values of $\beta$ around
the critical value $\beta_C$. For $\beta<\beta_C$ the solution has a minimum
(blue curve).
For $\beta = \beta_{\rm C}$, the function $v(x)$ vanishes
at $x_0=2\lambda_0$, and its support is $x \in [\lambda_0, 2\lambda_0]$. 
This is shown as the
black curve in Figure~\ref{Fig:vcr}. For $\beta > \beta_{\rm C}$ the
solution vanishes at $x_1 < 2\lambda_0$ (red curve).

\begin{figure}
    \centering
   \includegraphics[width=2.8in]{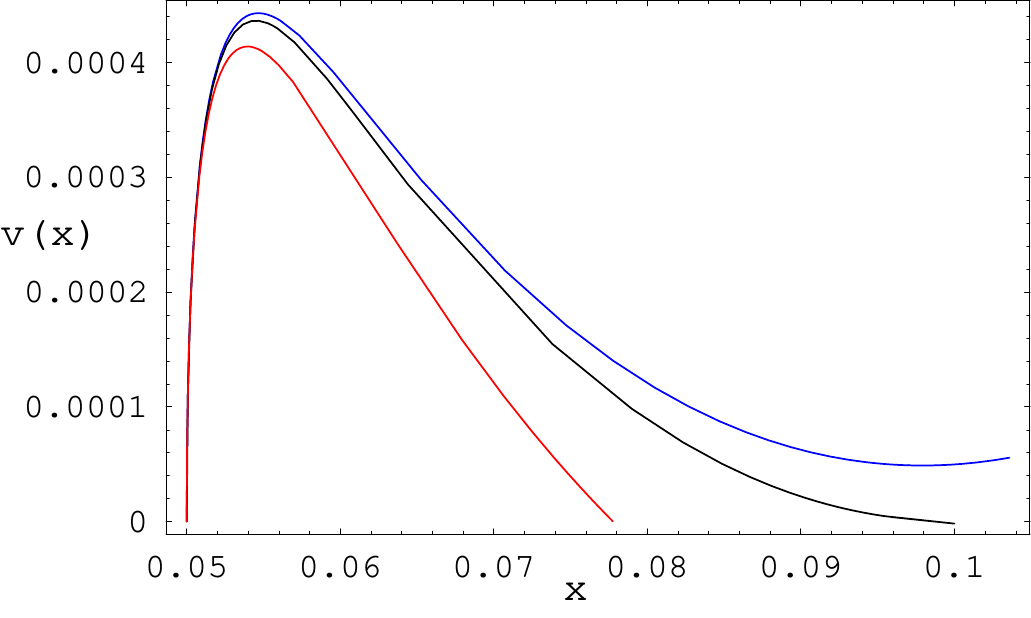}
    \caption{ Numerical solutions for $v(x)$ for  
$\beta<\beta_C$ (blue), $\beta=\beta_C$ (black)
and $\beta > \beta_C$ (red). The plots correspond to $\sigma=0.2, \lambda_0=0.05$
and $\beta=0.0625$ (blue),  $\beta=\beta_C=0.063$ (black) and $\beta=0.066$
(red).  }
\label{Fig:vcr}
 \end{figure}

\section{Eurodollar futures explosion}
\label{sec:6}

Consider the pricing of an Eurodollar futures contract on the rate $L(T,T+\delta)$
in the log-normal quasi-Gaussian HJM model. This requires the calculation of the
expectation of this rate $\mathbb{E}^{\mathbb{Q}}[L(T,T+\delta)]$ 
in the risk-neutral measure \cite{HK}.
Since the rate $L(T_1,T_2)$ is related to the zero coupon bond $P(T_1,T_2)$ as
$L(T_1,T_2) = \frac{1}{T_2-T_1}(P^{-1}(T_1,T_2) - 1)$, see e.g. \cite{APbook}, 
the pricing of the Eurodollar futures contract can be reduced to the 
evaluation  of the expectation value of the inverse zero coupon bond 
$P^{-1}(T,T+\delta)$ in the risk-neutral measure.

Using Eq.~(\ref{ZCB}) for the zero coupon bond price $P(t,T)$ we get
\begin{equation}
\mathbb{E}^{\mathbb{Q}}[P^{-1}(T,T+\delta)] = \frac{P(0,T)}{P(0,T+\delta)}
\mathbb{E}^{\mathbb{Q}}\left[\exp\left(G(T,T+\delta) x_T + \frac12 G^2(T,T+\delta) y_T \right)\right]\,.
\end{equation}
This expectation is bounded from below by the Jensen inequality as
\begin{align}
\mathbb{E}^{\mathbb{Q}}\left[\exp\left(G(T,T+\delta) x_T + \frac12 G^2(T,T+\delta) y_T \right)\right] 
&\geq
\exp\left(G(T,T+\delta) \mathbb{E}^{\mathbb{Q}}[ x_T] + \frac12 G^2(T,T+\delta) \mathbb{E}^{\mathbb{Q}}[y_T] \right) 
\\
&\geq
\exp\left(G(T,T+\delta) \mathbb{E}^{\mathbb{Q}}[ x_T]\right)\,,\nonumber
\end{align}
where we recall that $x_t=r_t-\lambda(t)$.
We used in the last inequality the positivity of $y_t$. 
For the small-noise limit, and sufficiently small $\beta$, the expectation on
the right-hand side explodes to infinity since we have 
$\lim_{t\to \tau_\infty} \mathbb{E}^{\mathbb{Q}}[ r_t]\simeq 
\lim_{t\to \tau_\infty} r(t)=\infty$.
We conclude thus that Eurodollar futures contract prices in this model must 
explode before the explosion time $\tau_\infty$.

The pricing of Eurodollar futures in the quasi-Gaussian HJM was considered in
\cite{CaZhu}, under several volatility specifications: normal, log-normal 
(identical to the model considered here) and square-root. The model was 
solved by finite differences methods. 
The numerical tests in \cite{CaZhu} assumed a flat yield curve at 
$\lambda_0=0.05$. Several values of the rates volatility were considered
$\sigma = 0.05,0.1,0.3$, and mean-reversion $\beta_1=0.01$ and $\beta_2=0.1$.
For all these cases one has $\beta_C \leq 0.1$. Typical values of the 
mean-reversion
parameter in fixed-income markets are $\beta \leq 0.1$, so the condition 
$\beta < \beta_C$ is satisfied unless $\sigma,\lambda_0$ are very small,
such that $\sigma \sqrt{\lambda_0} \leq 0.07$. 

Assuming $\beta_1 \simeq 0$, Proposition~\ref{prop:beta0} gives explosion 
times less than or equal to 266Y, 133Y and 44Y, respectively. These are much 
larger than typical maturities of Eurodollar futures contracts, which are 
listed on exchanges up to 10Y, although most of the liquid quotes are within 7Y.

\section{Summary and conclusions}

We studied the qualitative behavior of the solutions of the log-normal 
quasi-Gaussian HJM model using a deterministic approximation, which 
corresponds to the small Brownian noise limit. Under this approximation, 
we showed that the short rate may explode to infinity
in a finite time. This is relevant for the simulation and use of the model
for pricing financial derivatives. 

The small-noise solutions can be used to guide the construction of finite 
difference or tree approximations of the model SDE. The explosion phenomenon 
implies that Eurodollar futures prices may become infinite under certain
conditions on maturity and model parameters. Explicit explosion criteria are
presented, which give an upper bound on the explosion time, under weak 
assumptions on the shape of the initial yield curve.

\appendix

\section*{Appendix: A differential inequality}

We prove in this Appendix that the solutions of the differential inequality
(\ref{ODEfull3}) are bounded from 
below by the solution of the equation obtained by taking replacing the
inequality sign with equality.
Thus a sufficient condition for the explosions of the 
solution of (\ref{ODEfull}) under Assumption~\ref{assump1} is that the 
solution with $\Lambda(t)=0$ has an explosion. On the other hand, the 
absence of an explosion for the solution of the latter equation does not say 
anything about the presence or absence of an explosion in the former case.

\begin{lemma}\label{lemma:1}
Denote $R(t)$ the solution of the ODE obtained by replacing the inequality 
sign in (\ref{ODEfull3}) with equality, and satisfying the initial conditions 
(\ref{ic}). The solutions
of the differential inequality are bounded from below by $R(t)$
\begin{equation}
r(t) \geq R(t)\,, \qquad t \geq 0\,.
\end{equation}
\end{lemma}

\begin{proof}
Write $r(t) = R(t) + \delta(t)$ where $R(t)$ is the solution of the ODE. 
Taking differences gives a differential inequality for $\delta(t)$ 
\begin{equation}
\delta''(t) + 3\beta \delta'(t) + 2\beta^2 \delta(t) \geq 2\sigma^2 R(t) \delta(t) +\sigma^2 \delta(t)^2 \geq 2\sigma^2 R(t) \delta(t)\,.
\end{equation}
with initial conditions $\delta(0) = 0\,,\delta'(0)=0$.

Define $\varepsilon(t)$ as
$\delta(t) = e^{-\frac32 \beta t} \varepsilon(t)$, with initial conditions
$\varepsilon(0) = 0 \,, \varepsilon'(0) = 0 $.
This function satisfies the inequality
\begin{equation}\label{ineq1}
\varepsilon''(t) \geq F(t) \varepsilon(t) =
\left( \frac14\beta^2 e^{\frac34\beta t} + 2\sigma^2 R(t) \right) \varepsilon(t),
\end{equation}
where we defined
$F(t) = \frac14\beta^2 e^{\frac34\beta t} + 2\sigma^2 R(t)$.
%This is a positive function for $t>0$.

The inequality (\ref{ineq1}) gives
\begin{equation}
\varepsilon'(t) = \int_0^t \varepsilon''(s) ds 
\geq \int_0^t F(s) \varepsilon(s) ds.
\end{equation}
Integrating again over $t$ we get
\begin{equation}\label{ineq}
\varepsilon(t) = 
\int_0^t \varepsilon'(u) du \geq \int_0^t \left( \int_0^u F(s) 
\varepsilon(s) ds \right) du = 
\int_0^t (t-s) F(s) \varepsilon(s) ds,
\end{equation}
where we exchanged the order of integration in the last step.
Note that $(t-s)F(s)$ is non-negative and continuous function for $0\leq s\leq t$.
Therefore by the generalized Gronwall's inequality in \cite{CM}, we have $\varepsilon(t)\geq 0$.
\end{proof}

\end{document}